\documentclass{article}[12pt,a4]
\pdfoutput=1

\usepackage[T1]{fontenc}
\usepackage{fullpage}
\usepackage[english]{babel}
\usepackage{authblk}
\usepackage{amsmath}
\usepackage[utf8]{inputenc}
\usepackage{amssymb}
\usepackage{amsthm}
\usepackage{mathtools}
\usepackage{xspace}
\usepackage[hyperfootnotes=false,breaklinks,bookmarks=false]{hyperref}
\usepackage{algorithmic}
\usepackage{algorithm}
\usepackage{color}
\usepackage{graphicx}

\newtheorem{definition}{Definition}[section]

\newtheorem{lemma}[definition]{Lemma}
\newtheorem{theorem}[definition]{Theorem}

\newcommand{\bigo}{\mathcal{O}}
\newcommand{\etal}{{et~al.}\xspace}
\newcommand{\cost}{\delta}

\title{A note on distance labeling in planar graphs}

\author[1]{Paweł Gawrychowski}
\author[2]{Przemysław Uznański}

\affil[1]{University of Haifa, Israel and University of Wrocław, Poland}
\affil[2]{ETH Z\"urich, Switzerland}

\date{}
\begin{document}
\maketitle

\begin{abstract}
A distance labeling scheme is an assignments of labels, that is binary strings, to all nodes of a graph, so that the 
distance between any two nodes can be computed from their labels and the labels are as short as possible. A major 
open problem is to determine the complexity of distance labeling in unweighted and undirected planar graphs. It is 
known that, in such a graph on $n$ nodes, some labels must consist of $\Omega(n^{1/3})$ bits, but the best known 
labeling scheme uses labels of length $\bigo(\sqrt{n}\log n)$ [Gavoille, Peleg, P\'erennes, and Raz, J.\ Algorithms, 
2004]. We show that, in fact, labels of length $\bigo(\sqrt{n})$ are enough.
\end{abstract}

\vfill

\section{Introduction}

A distance labeling scheme is an assignments of labels, that is binary strings, to all nodes of a graph $G$, so that
the distance $\cost_{G}(u,v)$ between any two nodes $u,v$ can be computed from their labels. The main goal is
to make the labels as short as possible, that is to minimize the maximum length of a label. The secondary goal
is to optimize the query time, that is the time necessary to compute $\cost_{G}(u,v)$ given the labels of $u$ and
$v$. We consider distance labeling for unweighted undirected graphs on $n$ nodes. This was first 
considered by Graham and Pollak~\cite{pollak}, who obtained labels of length $\bigo(n)$. The decoding time 
was subsequently improved to $\bigo(\log\log n)$ by Gavoille \etal~\cite{DBLP:journals/jal/GavoillePPR04},
then to $\bigo(\log^{*}n)$ by Weimann and Peleg~\cite{WP11}, and finally Alstrup
\etal~\cite{DBLP:conf/soda/AlstrupGHP16} obtained $\bigo(1)$ decoding time with labels of length
$\frac{\log 3}{2}n+o(n)$. It is known that, in the general case, some labels must consist of at least $\frac{n}{2}$
bits~\cite{moon1965minimal,KannanNR92}, so achieving sublinear bounds is not possible.

Better schemes for distance labeling are known for restricted classes of graphs. As a prime example, trees
admit distance labeling scheme with labels of length $\frac{1}{4}\log^{2}n+o(\log^{2}n)$
bits~\cite{FreedmanGNW16}, and this is known to be tight up to lower order terms~\cite{alstrup2015distance}.
In fact, any sparse graph admits a sublinear distance labeling scheme~\cite{Sublinear} (see also
\cite{GawrychowskiKU16} for a somewhat simpler construction). However, the best known upper bound
is still rather far away from the best known lower bound of $\Omega(\sqrt{n})$.
An important subclass of sparse
graphs are planar graphs, for which Gavoille \etal~\cite{DBLP:journals/jal/GavoillePPR04} constructed
labeling schemes of length $\bigo(\sqrt{n}\log n)$.
They also proved that in any such scheme some label must consist of
$\Omega(n^{1/3})$ bits. Their upper bound of $\bigo(\sqrt{n}\log n)$ bits is also valid for weighted undirected
planar graphs if the weights are polynomially bounded. Very recently, Abboud and Dahlgaard~\cite{AbboudD16}
showed a matching lower bound of $\Omega(\sqrt{n}\log n)$ in such a setting. However, determining the
complexity of distance labeling in unweighted planar graphs remains to be a major open problem in this
area.

We design a better distance labeling scheme for unweighted undirected planar graphs, where the labels
consist of only $\bigo(\sqrt{n})$ bits. While this is only a logarithmic improvement, we believe lack of any
progress in the last 12 years makes any asymptotic decrease desirable. Our improvement is based on
a tailored version of the well-known planar separator lemma. We explain the idea in more detail after presenting
the previous labeling scheme. 

\section{Previous scheme}

We briefly recap the scheme of Gavoille \etal~\cite{DBLP:journals/jal/GavoillePPR04}. Their construction is
based on the notion of separators, that is, sets of vertices which can
be removed from the graph so that every remaining connected component consists of at most $\frac{2}{3}n$
nodes. By the classical result of Lipton and Tarjan~\cite{LiptonTarjan} any planar graph on $n$ nodes
has such a separator consisting of $\bigo(\sqrt{n})$ nodes. Now the whole construction proceeds as follows:
find a separator $S$ of the graph, and let $G_1,G_2,\ldots$ be the connected components of $G \setminus S$.
The label of $v \in G_i$ in $G$, denoted $\ell_G(v)$, is composed of $i$, $\ell_{G_i}(v)$ and the
distances $\cost_{G}(v,u)$ for all $u\in S$ written down in the same order for every $v\in G$.
A label of $v\in S$ consist of only the distances $\cost_{G}(v,u)$ for all $u\in S$.

The space complexity of the whole scheme is dominated by the space required to store $|S|$
distances, each consisting of $\log n$ bits, so $\bigo(\sqrt{n}\log n)$ bits in total.
(All logarithms are in base 2.) The bound of $\bigo(\sqrt{n})$ on the size of a separator is asymptotically tight.
However, the total length of the label of $v\in G$ (in bits) depends not on the size of the separator,
but on the number of bits necessary to encode the distances from $v$ to the nodes of the separator.
If the separator is a simple cycle $(u_{1},u_{2},\ldots,u_{|S|})$ then $|\cost(v,u_{i})-\cost(v,u_{i+1})|\leq 1$,
for every $i=1,2,\ldots,|S|-1$, and consequently writing down $\cost(v,u_{1})$ explicitly and then storing all the
differences $\cost(v,u_{i})-\cost(v,u_{i+1})$ takes only $\bigo(\sqrt{n})$ bits in total. It is known that
if the graph is triangulated, there always exists a simple cycle separator~\cite{MILLER1986265},
so for such graphs labels of length $\bigo(\sqrt{n})$ are enough.
We will show that, in fact, for any planar graph it is possible to select a separator so that the obtained
sequence of differences is compressible. This will be done by triangulating the faces using 
carefully designed gadgets.

\section{Improved scheme}

We introduce the notion of weighted separators. Consider a planar graph, where every node has a
non-negative weight and all these weights sum up to 1. Then a set of nodes is a weighted separator
if after removing these nodes the total weight of every remaining connected component is at most $\frac{2}{3}$.
We have the following well-known theorem (the result is in fact more general and
allows assigning weights also to edges and faces, but this is not needed in our application):

\begin{theorem}[see~\cite{MILLER1986265}]
\label{thm:separator}
For every 2-connected planar graph on $n$ nodes having assigned nonnegative weights summing
up to 1, there exists a simple cycle of length at most $2\sqrt{2 d  n}$ which is a
weighted separator, where $d$  is the maximum face size.
\end{theorem}

Now we are ready to use this tool to show the main technical lemma of this section:

\begin{lemma}
\label{lem:logseparator}
Any planar graph $G$ has a separator $S$, such that $$\sum_{i} \log \cost_{G}(u_i,u_{i+1}) = \bigo(\sqrt{n})$$ for some ordering $u_{1},u_{2}\ldots,u_{|S|}$ of all vertices of $S$.
\end{lemma}

Before proving the lemma, we first describe a family of \emph{subdivided cycles}. A subdivided cycle on $s\geq 3$ nodes, denoted
$D_s$, consists of a cycle $C_s=(v_1,\ldots,v_s)$ and possibly some auxiliary nodes. $D_3$
and $D_4$ is simply $C_3$ and $C_4$, respectively.  For $s>4$, we add $\lceil \frac{s}{2} \rceil$
auxiliary vertices $u_1,\ldots,u_{\lceil \frac{s}{2} \rceil}$, and connect every $v_i$ with
$u_{\lceil \frac{i}{2} \rceil}$. To complete the construction, we recursively build $D_{\lceil \frac{s}{2} \rceil}$
and identify its cycle with $(u_1,\ldots,u_{\lceil \frac{s}{2} \rceil})$. (An example of such
a subdivided cycle on 10 nodes is shown in Figure~\ref{fig:subd}.) We have the following property.

\begin{lemma}
\label{lem:subdivided}
For any $u,v \in C_s$, $\cost_{D_s}(u,v) \ge \log(1+\cost_{C_s}(u,v))$.
\end{lemma}

\begin{figure}[t]
\centering\includegraphics[width=0.8\textwidth]{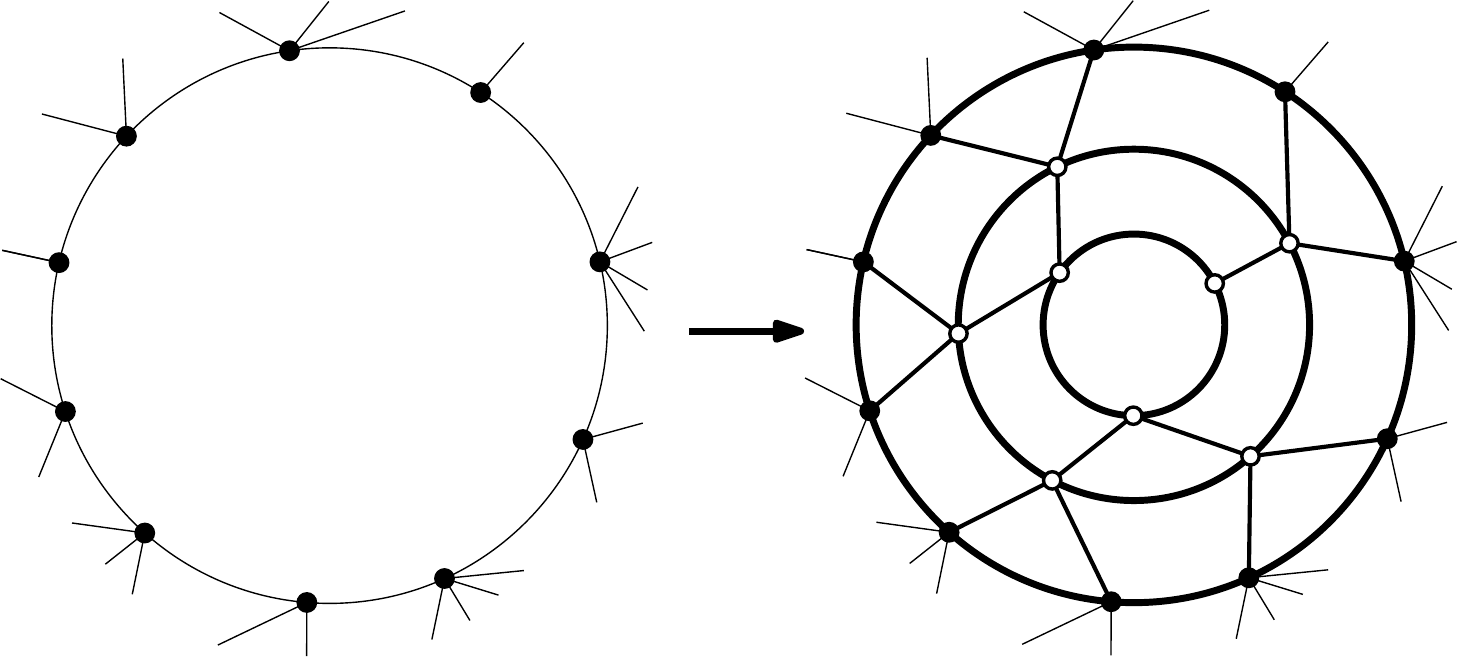}
\caption{A face of size 10 is transformed by replacing $C_{10}$ with $D_{10}$ containing 8 new auxiliary nodes.}
\label{fig:subd}
\end{figure}

\begin{proof}
We apply induction on $s$. It is easy to check that the lemma holds when $s \leq 4$, so we assume $s \geq 5$.
Let us denote $\cost_{D_{s}}(u,v)=d'$ and $\cost_{C_{s}}(u,v)=d$.
We proceed with another induction on $d$. It is easy to check that the lemma holds when $d \le 2$, so we
assume $d \geq 3$. Consider a shortest path connecting $u$ and $v$ in $D_{s}$. If it
consists of only auxiliary nodes except for the endpoints $u$ and $v$, then we consider the immediate
neighbors of $u$ and $v$ on the path, denoted $u'$ and $v'$, respectively. Since $u'$ and $v'$ must
belong to the cycle of $D_{\lceil \frac{s}{2} \rceil}$ and the distance between them in the corresponding
$C_{\lceil \frac{s}{2} \rceil}$
is at least $\lfloor \frac{d}{2} \rfloor$, by the inductive assumption:
$$d' \ge 2 +  \log (1+\lfloor d/2 \rfloor)  = \log (4+4\lfloor d/2 \rfloor ) \ge  \log(1+d).$$
Otherwise, let $u_{0}=u,u_{1},u_{2},\ldots,u_{s}=v$ be all nodes of the cycle $C_{s}$ appearing on the
path, where $s\geq 2$. We have that:
$$d'=\sum_{i} \cost_{D_{s}}(u_{i},u_{i+1})$$
and by the inductive assumption:
$$d'=\sum_{i} \cost_{D_{s}}(u_{i},u_{i+1}) \ge \sum_{i} \log (1+\cost_{C_{s}}(u_{i},u_{i+1})).$$
Let us denote $d_{i} = \cost_{C_{s}}(u_{i},u_{i+1}) \geq 1$ (as otherwise the path could have been shortened).
By the triangle inequality, $\sum_{i} 1+d_{i} \geq 1+d$,
so also $\prod_{i} 1+d_{i} \geq 1+d$. Therefore, $d' \geq \sum_{i} \log (1+d_{i}) = \log \prod_{i} 1+d_{i} \geq \log(1+ d)$.
\end{proof}

\begin{proof}[Proof of Lemma~\ref{lem:logseparator}]
Let $G'$ be the graph constructed from $G$ by replacing every face (including the external face) with
a subdivided cycle of appropriate size.
More precisely, let $(v_{1},v_{2},\ldots,v_{s})$ be a cyclic walk of a face of $G$. Note that nodes $v_{i}$ are not
necessarily distinct. We create a subdivided cycle $D_{s}$ and identify its cycle $C_{s}$ with
$(v_{1},v_{2},\ldots,v_{s})$.
Clearly $G'$ is also planar and each of its faces is either a triangle or a square. Additionally, 
every subdivided cycle is 2-connected, so also the whole $G'$ is 2-connected.
Since any subdivided cycle has at most as many auxiliary vertices as cycle vertices and lengths of all cyclic walks sum up to twice the number of edges, which is at most
$3n-6$, $G'$ contains at most $n' = n+2\cdot (3n-6) < 7n$ vertices.

We assign weights to vertices of $G'$ in such a way that that every vertex that also appears in $G$ has
weight 1 and every new vertex has weight 0. By Theorem~\ref{thm:separator} there exists a weighted simple cycle
separator $S'$ in $G'$ of size at most $2\sqrt{28 n}$.
Then $S=S' \cap G$ is a separator in $G$.
Because $S'$ is a simple cycle separator, $S=(u_{1},u_{2}\ldots,u_{c})$,
and $u_{i}$ and $u_{i+1}$ are incident to the same face of $G$ that has been replaced with a subdivided
cycle $D_{s_{i}}$ such that $S'$ connects $u_{i}$ and $u_{i+1}$ either directly or by visiting some
auxiliary nodes of $D_{s_{i}}$, for every $i=1,2,\ldots,c$ (we assume $u_{c+1}=u_{1}$).
Let $v_{i}$ and $v'_{i}$ denote nodes of $D_{s_{i}}$
that have been identified with $u_{i}$ and $u_{i+1}$, respectively. Then:
$$\sum_i \cost_{D_{s_{i}}}(v_i,v'_{i}) \le |S| = \bigo(\sqrt{n}).$$
By Lemma~\ref{lem:subdivided}, $\cost_{D_{s_{i}}}(v_i,v'_{i}) \ge \log \cost_{C_{s_{i}}}(v_i,v'_{i})$, so:
$$\sum_i \log \cost_{G}(v_i,v_{i+1}) \le \sum_i \log \cost_{C_{s_{i}}}(v_i,v'_{i}) \le \bigo(\sqrt{n})$$
as required.
\end{proof}

Now we proceed to the main result of this section:
\begin{theorem}
\label{thm:planar}
Any planar graph on $n$ nodes admits a distance labeling scheme of size $\bigo(\sqrt{n})$.
\end{theorem}
\begin{proof}
We proceed as in the previously known scheme of size $\bigo(\sqrt{n}\log n)$, except that in every
step we apply our Lemma~\ref{lem:logseparator}. In more detail, to construct the label of every
$v\in G$ we proceed as follows. First, we find a separator $S=\{u_{1},u_{2}\ldots,u_{c}\}$ using Lemma~\ref{lem:logseparator}. For every $v\in G$ we encode its distances to all nodes of the separator.
This is done by first writing down $\cost_{G}(v,u_1)$ explicitly, and then $\cost_{G}(v,u_{i})-\cost_{G}(v,u_{i-1})$
for $i=2,\ldots,c$. All these numbers are encoded in binary, so by
$|\cost_{G}(v,u_{i})-\cost_{G}(v,u_{i-1})|\leq \cost_{G}(u_{i-1},u_{i})$ and the properties of our separator
this encoding takes $\bigo(\sqrt{n})$ bits in total. Second, for every node we store the name of its connected
component of $G\setminus S$. Third, we recurse on every connected component of $G\setminus S$
and append the obtained labels to the current labels. To calculate $\cost_{G}(v,v')$, we
proceed recursively: we compute $\min_{u\in S} (\cost_{G}(v,u)+\cost_{G}(v',u))$ and then,
if $v$ and $v'$ belong to the same connected component of $G\setminus S$, proceed recursively.
The correctness is clear:
either a shortest path between $u$ and $v$ is fully within one of the connected components,
or it visits some $u\in S$, and in such case we can recover $\cost_{G}(v,u)+\cost_{G}(v',u)$ from the stored
distances. The final size of every label is $\bigo(\sqrt{n}+\sqrt{\frac{2}{3}n}+\ldots)=\bigo(\sqrt{n})$ bits.
\end{proof}

\bibliography{biblio}
\end{document}